\documentclass[envcountsame]{llncs} 
\usepackage{ifthen}
\newcommand{\techRep}{true} 
\newcommand{\iftechrep}{\ifthenelse{\equal{\techRep}{true}}}
\usepackage{amsmath}
\usepackage{amssymb}
\usepackage{hyperref}
\usepackage{microtype} 
\usepackage{rgalg}
\usepackage{tikz}
\usetikzlibrary{automata,arrows,calc,positioning}
\usepackage{xcolor}
\usepackage{wrapfig}

\definecolor{darkblue}{rgb}{0,0,0.4}
\hypersetup{colorlinks,linkcolor=darkblue,citecolor=darkblue,urlcolor=darkblue}
\hypersetup{
  pdftitle={Tree Buffers},
  pdfauthor={Radu Grigore and Stefan Kiefer}}

\newenvironment{qtheorem}[1]{%
{\medskip\noindent\bfseries Theorem #1.\hspace{0mm}}
\begin{itshape}%
}{%
\end{itshape}%
\medskip
}

\tikzset{irrelevant transition/.style={->,>=stealth',thin}}
\tikzset{relevant transition/.style={->,>=stealth',very thick}}

\newcommand{\A}{\mathcal{A}}%
\newcommand{\Hs}[1]{H^{(#1)}}%
\newcommand{\Active}{\mathit{Active}}%
\newcommand{\Nodes}{\mathit{Nodes}}%
\newcommand{\Rs}[1]{R^{(#1)}}%
\newcommand{\Ms}[1]{M^{(#1)}}%
\newcommand{\addchild}{\textsc{add\_child}}%
\newcommand{\children}{\mathit{children}}%
\newcommand{\cnt}{\mathit{cnt}}%
\newcommand{\cutParent}{\textsc{cut\_parent}}%
\newcommand{\depth}{\mathit{depth}}%
\newcommand{\deque}{\textsc{deque}}%
\newcommand{\df}{\emph}
\newcommand{\enque}{\textsc{enque}}%
\newcommand{\expand}{\textsc{expand}}%
\newcommand{\deactivate}{\textsc{deactivate}}%
\newcommand{\gc}{\textsc{gc}}%
\newcommand{\history}{\textsc{history}}%
\newcommand{\initialize}{\textsc{initialize}}%
\newcommand{\level}{\mathit{level}}%
\newcommand{\mem}{\mathit{mem}}%
\newcommand{\mems}[1]{\mathit{mem}^{(#1)}}%
\newcommand{\memOld}{\mathit{memOld}}%
\newcommand{\parent}{\mathit{parent}}%
\newcommand{\pseudocodesize}{\small}%
\newcommand{\queue}{\mathit{queue}}%
\newcommand{\rep}{\mathit{rep}}%
\newcommand{\samos}[1]{s_{\sf amo}^{(#1)}}%
\newcommand{\srts}[1]{s_{\sf rt}^{(#1)}}%
\newcommand{\sgcs}[1]{s_{\sf gc}^{(#1)}}%
\newcommand{\tran}[2]{%
  \mathrel{{%
    \setbox0=\hbox{$\scriptstyle #2$}\dimen0=\wd0 \advance\dimen0 by .75em %
    \tikz[baseline]%
      \draw[#1]%
         (0,0.5ex) -- node[above]
          {\vbox to .2ex{\vss\hbox{$\scriptstyle #2$}}} (\the\dimen0,0.5ex)%
    ;%
  }}%
}
\newcommand{\itran}[1]{\tran{irrelevant transition}{#1}}
\newcommand{\rtran}[1]{\tran{relevant transition}{#1}}
\newcommand{\ini}{\rtran{\ }\!1}%
\tikzset{initial text={}}
\tikzset{every initial by arrow/.style={relevant transition}}

\title{Tree Buffers} 
\author{Radu Grigore \and Stefan Kiefer
}
\institute{University of Oxford}

\begin{document} 
\maketitle
\begin{abstract} 
In \emph{runtime verification}, the central problem
  is to decide if a given program execution violates a given property.
In \emph{online} runtime verification, a monitor observes a program's execution as it happens.
If the program being observed has hard real-time constraints, then the monitor inherits them.
In the presence of hard real-time constraints it becomes a challenge to maintain enough
  information to produce \emph{error traces}, should a property violation be observed.
In this paper we introduce a data structure, called \emph{tree buffer},
 that solves this problem in the context of automata-based monitors:
If the monitor itself respects hard real-time constraints,
  then enriching it by tree buffers makes it possible to provide error traces,
  which are essential for diagnosing defects.
We show that tree buffers are also useful in other application domains.
For example, they can be used to implement functionality of \emph{capturing groups} in regular expressions.
We prove optimal asymptotic bounds for our data structure,
 and validate them using empirical data from two sources:
  regular expression searching through Wikipedia,
  and runtime verification of execution traces obtained from the DaCapo test suite.
\end{abstract}
\section{Introduction} 
\label{sec-intro}

In runtime verification, a program is instrumented to emit events at certain times, such as method calls and returns.
A monitor runs in parallel, observes the stream of events, and identifies bad patterns.
Often, the monitor is specified by an automaton (for example, see \cite{%
  dblp:conf/fm/barringerfhrr12,%
  dblp:conf/rv/0002kv13,%
  dblp:journals/fmsd/finkbeiners04,%
  dblp:conf/isola/havelund14,%
  dblp:journals/corr/abs-1112-5761}).
When the accepting state of the automaton is reached, the last event of the program corresponds to a bug.
At this point, developers want to know how was the bug reached.
For example, the bug could be that an invalid iterator is used to access its underlying collection.
An iterator becomes invalid when its underlying collection is modified,
for instance by calling the \textsc{remove} method of another iterator for the same collection.
In order to diagnose the root cause of the bug, developers will want to determine how exactly the iterator became invalid.
Of particular interest will be an \emph{error trace}:
  the last few relevant events that led to a bug.
In the context of static verification,
  error traces have proved to be invaluable in diagnosing the root cause of bugs%
  ~\cite{dblp:journals/scp/leinoms05}.
However, runtime verification tools
  (such as \cite{%
    dblp:conf/rv/bodden11,%
    dblp:conf/icse/jinmlr12,%
    dblp:conf/rv/luozljmsr14})
  shy away from providing error traces,
  perhaps because adding this functionality would impact efficiency.
The goal of this paper is to provide the algorithmic foundations of efficient monitors that can provide error traces for a very general class of specifications.

\newcommand{\iter}{\mathit{iter}}%
\newcommand{\biter}{\pmb{iter}}%
\newcommand{\hasNext}{\mathit{hasNext}}%
\newcommand{\bhasNext}{\pmb{hasNext}}%
\newcommand{\nex}{\mathit{next}}%
\newcommand{\bnex}{\pmb{next}}%
\newcommand{\other}{\mathit{other}}%
\newcommand{\ba}{\boldsymbol{a}}%
\newcommand{\bb}{\boldsymbol{b}}%

\begin{figure}[tb] 
\begin{tabular}{cc}
\begin{tikzpicture}[scale=2] 
\tikzset{n/.style={draw,circle}}
\tikzset{ia/.style={irrelevant transition}}
\tikzset{ra/.style={relevant transition}}
\node[n] (1) at (1,0) {$1$};
\node[n] (2) at (2,0) {$2$};
\node[n] (3) at (2,0.8) {$3$};
\node[n,accepting] (4) at (3,0) {$4$};
\draw[ra] (0.5,0) -- (1);
\draw[ia] (1) edge[loop,out=250,in=290,looseness=15]  node[below] {$*$} (1); 
\draw[ra] (1) to node[above] {$\biter$} (2);
\draw[ia] (2) edge[loop,out=250,in=290,looseness=15]  node[below] {$\other$} (2);
\draw[ra] (2) to[bend left=20] node[left] {$\bhasNext$} (3);
\draw[ra] (3) to[bend left=20] node[right] {$\bnex$} (2);
\draw[ia] (3) edge[loop,out=200,in=160,looseness=15]  node[left] {$\other,\hasNext$} (3);
\draw[ra] (2) -- node[above] {$\bnex$} (4);
\draw[ia] (4) edge[loop,out=250,in=290,looseness=15]  node[below] {$*$} (4);
\end{tikzpicture}
& \quad
\begin{tikzpicture}[scale=2] 
\tikzset{n/.style={draw,circle}}
\tikzset{ia/.style={irrelevant transition}}
\tikzset{ra/.style={relevant transition}}
\node[n] (1) at (1,0) {$1$};
\node[n] (2) at (2,0) {$2$};
\node[n,accepting] (3) at (3,0) {$3$};
\draw[ra] (0.5,0) -- (1);
\draw[ra] (1) to[bend left=20] node[above] {$\ba$} (2);
\draw[ra] (2) to[bend left=20] node[below] {$\bb$} (1);
\draw[ra] (2) -- node[above] {$\bb$} (3);
\draw[ra] (1) edge[loop,out=110,in=70,looseness=15]  node[above] {$\ba$} (1); 
\draw[ia] (1) edge[loop,out=250,in=290,looseness=15]  node[below] {$b,c$} (1); 
\draw[ia] (2) edge[loop,out=250,in=290,looseness=15]  node[below] {$a,c$} (2);
\draw[ia] (3) edge[loop,out=250,in=290,looseness=15]  node[below] {$a,b,c$} (3);
\end{tikzpicture}
\\ (a) & (b)
\end{tabular}
\caption{Two automata with relevant transitions in boldface.}
\label{fig-intro-NFAs}
\end{figure}
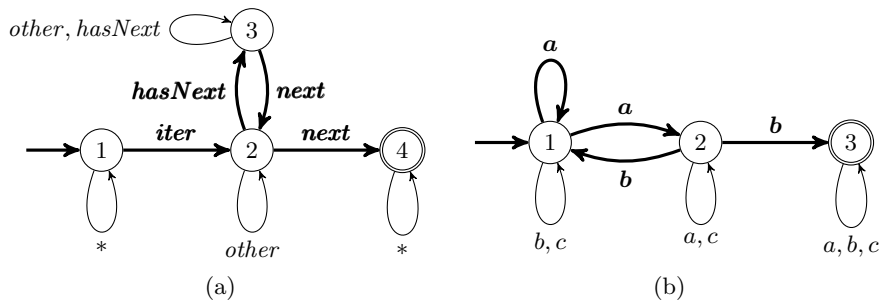 

Nondeterministic automata provide a convenient specification formalism for monitors.
They define both bugs and relevant events.
\autoref{fig-intro-NFAs}a shows an example automaton that specifies incorrect usage of an iterator:
it is a bug if an iterator is created (event~$\iter$),
and afterwards its \textsc{next()} method is called without a preceding call to \textsc{hasNext()}.
Throughout the paper we assume that the user specifies which transitions are \emph{relevant}.
In most applications, there is a natural way to choose the relevant transitions.
For example, in \autoref{fig-intro-NFAs}a and in many other runtime verification properties,
  the natural choice are the non-loop transitions.
Since the choice is natural, it can be automated;
  since the choice is dependent on application details, we do not focus on it.

We have to consider nondeterministic automata in general.
Nondeterministic finite automata allow exponentially more succinct specifications than deterministic finite automata.
In addition,
  in the runtime verification context we must use an automaton model
  that handles possibly infinite alphabets.
For most models of automata over infinite alphabets,
  the nondeterministic variant is strictly more expressive
  than the deterministic variant~\cite{%
    dblp:journals/tcs/bjorklunds10,%
    dblp:conf/focs/kaminskif90,%
    dblp:conf/fossacs/tzevelekosg13}.
Thus, we must consider nondeterminism not only to allow concise specifications,
  but also because some specifications cannot be defined otherwise.

Let us consider a concrete example:
  the automaton in \autoref{fig-intro-NFAs}b,
  consuming the stream of letters $cabbcab$.
(We say \df{stream} when we wish to emphasize that the elements of the sequence
  must be processed one by one, in an online fashion.)
One of the automaton computations labeled by $cabbcab$ is
$
 1 \itran{c} 1 \rtran{\ba} 1 \itran{b} 1 \itran{b} 1 \itran{c} 1 \rtran{\ba} 2 \rtran{\bb} 3
$,
  where relevant transitions are bold.
We say that the subsequence formed by the relevant transitions is an \df{error trace};
  here, $1\rtran{\ba}1\rtran{\ba}2\rtran{\bb}3$.

The main contribution of this paper is the design of a data structure that allows the monitor to do the following while reading a stream:
\begin{enumerate}
\item
The monitor keeps track of the states that the nondeterministic automaton could currently be in.
Whenever the automaton could be in an accepting state, the monitor reports (i) the occurrence of a bug,
and (ii) the last $h$ relevant transitions of a run that drove the automaton into an accepting state.
Here, $h$ is a positive integer constant that the user fixes upon initializing the monitor.
Due to the nondeterminism, a bug may have multiple such error traces,
 but the monitor needs to report only one of them.
\item
The monitor processes each event in a constant amount of \emph{time},
  thus paving the way for implementing real-time runtime verifiers that track error traces.
(There is a need for real-time verifiers~\cite{dblp:conf/rv/pikenw11}.)
Not only the time is constant, but also not much \emph{space} is wasted.
Wasted space occurs if the monitor keeps transitions that are not among the $h$~most recent relevant transitions.
\end{enumerate}
Due to the nondeterminism of the automaton,
those constraints force
the monitor to keep track of a \emph{tree} of computation histories.
For properties that can be monitored with \emph{slicing}~\cite{dblp:journals/corr/abs-1112-5761}
 the tree of computation histories has a very particular shape.
That shape allows for a relatively straightforward technique for providing error traces,
using linear buffers.
However, it has been shown that some interesting program properties,
  including \emph{taint} properties,
  cannot be expressed by slicing~\cite{%
    dblp:conf/fm/barringerfhrr12,%
    dblp:conf/tacas/grigoredpt13}.

In this paper we provide a monitor for \emph{general} nondeterministic automata,
at the same time
satisfying the properties 1~and~2 mentioned above.
The single most crucial step is the design of an efficient data structure, which we call \emph{tree buffer}.
A tree buffer operates on general trees and may be of independent interest.

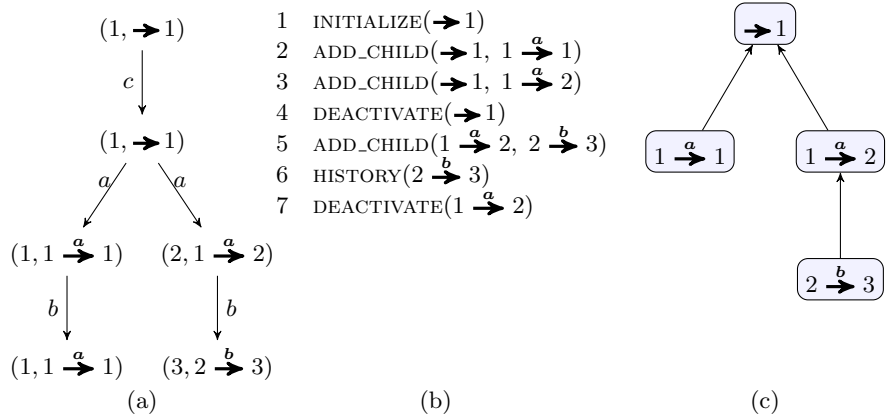
\begin{figure}[tb] 
\centering\small
\begin{tabular}{c@{\hskip-.25cm}c@{\hskip.5cm}c}
\begin{tikzpicture}[xscale=1,yscale=1.5,baseline=(current bounding box.north)]
\tikzset{n/.style={}}
\tikzset{ia/.style={irrelevant transition}}
\tikzset{ra/.style={relevant transition}}
\node[n] (0) at (0,0) {$(1, \ini)$};
\node[n] (1) at (0,-1) {$(1, \ini)$};
\node[n] (2a) at (-1,-2) {$(1, 1 \rtran{\ba} 1)$};
\node[n] (2b) at (+1,-2) {$(2, 1 \rtran{\ba} 2)$};
\node[n] (3a) at (-1,-3) {$(1, 1 \rtran{\ba} 1)$};
\node[n] (3b) at (+1,-3) {$(3, 2 \rtran{\bb} 3)$};
\draw[ia] (0) to node[left] {$c$} (1);
\draw[ia] (1) to node[above] {$a$} (2a);
\draw[ia] (1) to node[above] {$a$} (2b);
\draw[ia] (2a) to node[left] {$b$} (3a);
\draw[ia] (2b) to node[right] {$b$} (3b);
\end{tikzpicture}
&
\vtop{\vskip0pt
\begin{alg}
\0  $\initialize(\ini)$
\0  $\addchild(\ini,\;1\rtran{\ba}1)$
\0  $\addchild(\ini,\;1\rtran{\ba}2)$
\0  $\deactivate(\ini)$
\0  $\addchild(1\rtran{\ba}2,\; 2\rtran{\bb}3)$
\0  $\history(2\rtran{\bb}3)$
\0  $\deactivate(1\rtran{\ba}2)$
\end{alg}}
&
\begin{tikzpicture}[xscale=1,yscale=1.7,baseline=(current bounding box.north)]
\tikzset{n/.style={draw,rectangle,rounded corners,fill=blue!5}}
\tikzset{ia/.style={irrelevant transition}}
\tikzset{ra/.style={relevant transition}}
\node[n] (0) at (0,0) {$\ini$};
\node[n] (1a) at (-1,-1) {$1 \rtran{\ba} 1$};
\node[n] (1b) at (+1,-1) {$1 \rtran{\ba} 2$};
\node[n] (2b) at (+1,-2) {$2 \rtran{\bb} 3$};
\draw[ia] (1a) -- (0);
\draw[ia] (1b) -- (0);
\draw[ia] (2b) -- (1b);
\end{tikzpicture}
\\ (a) & (b) & (c)
\end{tabular}
\caption{Illustration of a monitor run of the automaton from \autoref{fig-intro-NFAs}b on the stream~$cab$.
Part~(a) shows the monitor's traversal of the automaton with some instrumentation.
Part~(b) shows the sequence of tree buffer operations that the monitor invokes.
Part~(c) shows the tree-buffer data structure that the monitor builds.}
\label{fig-intro-pseudocode}
\end{figure}

\smallskip \noindent \textbf{Tree Buffers for Monitoring. \quad}
A tree buffer is a data structure that stores parts of a tree.
Its two main operations are $\addchild(x,y)$, which adds to the tree a new node~$y$ as a child of node~$x$,
and $\history(x)$, which requests the $h$ ancestors of~$x$, where $h$ is a constant positive integer.
For memory efficiency the tree buffer distinguishes between \emph{active} and \emph{inactive} nodes.
When $\addchild(x,y)$ or $\history(x)$ is called, node~$x$ must be active.
In the case of~$\addchild(x,y)$, the new node~$y$ becomes active.
There is also a $\deactivate(x)$ operation with the obvious semantics.
One of the main contributions of this paper is the design of efficient algorithms that provide the functionality of tree buffers with asymptotically optimal time and space complexity.
More precisely, the $\addchild$ and $\deactivate$ operations take constant time, and the space wasted by nodes that are no longer accessible via $\history$ calls is bounded by a constant times the space occupied by nodes that \emph{are} accessible via $\history$ calls.

In the following, we give an example of how an efficient monitor operates,
  assuming that an efficient tree buffer is available.
Consider the automaton from \autoref{fig-intro-NFAs}b and the stream~$cab$.
The monitor keeps pairs of (1) a current automaton state~$q$, and of (2) a \emph{tree buffer node} with the most recent relevant transition of a run that led to~$q$.
Initially, this pair is $(1, \ini)$, as $1$ is the initial state of the automaton (see \autoref{fig-intro-pseudocode}).

Upon reading $c$,
  the automaton takes the transition $1\itran{c}1$,
  and the monitor simulates the automaton by evolving from $(1,\ini)$ to a new pair $(1,\ini)$:
  the first component remains unchanged because $1\itran{c}1$ is a loop;
  the second component remains unchanged because $1\itran{c}1$ is irrelevant.

Next, $a$~is read.
The automaton takes transitions $1\rtran{\ba}1$ and $1\rtran{\ba}2$, both relevant.
Corresponding to the automaton transition $1\rtran{\ba}1$,
  the monitor evolves $(1,\ini)$ into a new pair $(1,1\rtran{\ba}1)$:
  the first component remains unchanged because $1\rtran{\ba}1$ is a loop;
  the second component \emph{changes} because $1\rtran{\ba}1$ is \emph{relevant}.
Corresponding to the automaton transition $1\rtran{\ba}2$,
  the monitor \emph{also} evolves $(1,\ini)$ into a new pair $(2,1\rtran{\ba}2)$.
Now that two relevant transitions were taken,
  they are added to the tree buffer:
  both $1\rtran{\ba}1$ and $1\rtran{\ba}2$ are children of $\ini$.
Moreover, because $\ini$ is not anymore in any pair kept by the monitor,
  it is deactivated in the tree buffer.

Next, $b$~is read.
The automaton takes transitions $1\itran{b}1$,\; $2\rtran{\bb}1$, and $2\rtran{\bb}3$.
Out of the two transitions with the same target
  the monitor will pick only one to simulate, using an application specific heuristic.
In \autoref{fig-intro-pseudocode}, the monitor chose to ignore $2\rtran{\bb}1$.
Moreover,
  because $1\rtran{\ba}2$
    used to be in the monitor's pairs before $b$ was read but is not anymore,
  its corresponding tree buffer node is deactivated.
Finally, since state~$3$ is accepting,
  the monitor will ask the tree buffer for an error trace,
  by calling $\history(2\rtran{\bb}3)$.

In \autoref{fig:alg-nfa} we provide pseudocode formalizing the sketched algorithm.

\iftechrep{}{The full version of the paper~\cite{gk-cav15-tr} includes missing proofs and further details.}

\section{Tree Buffers} 

Consider a procedure that handles a stream of events.
At any point in time the procedure should be able to output the previous $h$ events in the stream,
 where $h$ is a fixed constant.
Such \emph{linear buffers} are ubiquitous in computer science,
 with applications, for example, in instruction pipelines~\cite{Smith88pipeline},
 voice-over-network protocols~\cite{Gunduzhan01voice}, and distributed operating systems~\cite{Kaashoek91DistrOS}.
Linear buffers can be easily implemented using \emph{circular buffers}, using $\Theta(h)$ memory and constant update time, which is clearly optimal.

While this buffering approach is simple and efficient, it is less appropriate if the streamed data is organized \emph{hierarchically}.
Consider a stream of events, each of which contains a link to one of the previous events.
We already saw an example of how such streams arise in runtime verification (\autoref{fig-intro-pseudocode}).
But, there are many other situations where such streams could arise;
  for example,
  when trees such as XML data are transmitted over a network, or
  when recording the spawned processes of a parallel computation, or
  when recording Internet browsing history.

A natural requirement for a buffer is to store the most recent data.
For a tree this could mean, for example, the leaves of the tree, or the $h$ ancestors of each leaf, where $h$ is a constant.
Observe that a linear buffer does not satisfy such requirements, because an old leaf or the parent of a new leaf may have been streamed much earlier, so that they have been removed from the buffer already.

A \df{tree buffer} is a tree-like data structure that satisfies such requirements.
It supports the following operations:
\begin{itemize}
\item $\initialize(x)$ initializes the tree with the single node~$x$ and makes $x$ active
\item $\addchild(x, y)$ adds node~$y$ as a child of the active node~$x$ and makes $y$ active
\item $\deactivate(x)$ makes $x$ inactive
\item $\expand(x, \{y_1,\ldots,y_n\})$ adds nodes $y_1,\ldots,y_n$ as children of the active node~$x$, makes $x$~inactive, and makes $y_1,\ldots,y_n$ active
\item $\history(x)$ requests the $h$~ancestors of the active node~$x$, where $h$~is a constant positive integer
\end{itemize}
A simple use case of a tree buffer consists of an $\initialize$ operation, followed by $\expand$ operations with $n>0$.
In this case the active nodes are always exactly the leaves.

\begin{figure}[tb]\pseudocodesize 
\hbox to \hsize{%
\hfil\vtop{%
\begin{alg}
\^  $\proc{initialize}(x)$
\0  $\parent(x):={\sf nil}$
\0  $\children(x):=0$
\0  $\Nodes:=\{x\}$
\0  $\Active:=\{x\}$
\0  $\mem:=1$
\0  $\memOld:=1$
\end{alg}
\medskip
\begin{alg}
\^  $\proc{add\_child}(x, y)$
\0  ~assert~ that $x\in\Active$ and $y \not\in\Nodes$
\0  $\parent(y):=x$
\0  $\children(x):=\children(x)+1$
\0  $\Nodes:=\Nodes\cup\{y\}$
\0  $\Active:=\Active\cup\{y\}$
\end{alg}}%
\hfil\vtop{%
\begin{alg}
\^  $\proc{deactivate}(x)$
\0  $\Active:=\Active-\{x\}$
\end{alg}
\medskip
\begin{alg}
\^  $\proc{history}(x)$
\0  ~assert~ that $x\in\Active$
\0  ${\it xs} := []$
\0  ~repeat~ $h$ times, or until $x={\sf nil}$
\1    ${\it xs} := x \cdot {\it xs}$
\1    $x := \parent(x)$
\0  ~return~ ${\it xs}$
\end{alg}
\medskip
\begin{alg}
\^  $\proc{expand}(x, \{y_1,\ldots,y_n\})$
\0  ~for~ $i\in\{1,\ldots,n\}$
\1     $\proc{add\_child}(x,y_i)$
\0  $\proc{deactivate}(x)$
\end{alg}}%
\hfil}
\caption{
  The {\sf naive} algorithm.
}\label{fig:alg-naive}
\end{figure}

The functionality of tree buffers is defined by the \textsf{naive} algorithm shown in \autoref{fig:alg-naive}.
The notation $f(x)$ stands for the field~$f$ of the node~$x$,
  while the notation $\textsc{f}(x)$
  stands for a call to function~$\textsc{f}$ with argument~$x$.
The field $\children$ and the variables $\mem$ and $\memOld$
  do not affect the behavior of the {\sf naive} algorithm:
  they are used later.
The assertions at the beginning of $\addchild$ and $\history$
  detect sequences of operations that are invalid.
For example, any sequence that does not start with a call to $\initialize$ is invalid.
For such invalid sequences,
  tree buffer implementations are not required to behave like the {\sf naive} algorithm.
For valid sequences we require implementations to be functionally equivalent,
 albeit performance is allowed to be different.

The \textsf{naive} algorithm is time optimal:
 $\initialize$, $\addchild$, and $\deactivate$ all take constant time;
 and $\history$ takes $O(h)$~time.
However, it is not space efficient, as it does not take advantage of $\deactivate$ operations: it does not delete nodes that are out of reach of $\history$.
The challenge in designing tree buffers lies in preserving both time and space efficiency.
On the one hand,
  it is not space efficient to store the whole tree.
On the other hand,
  it is not time efficient to exactly identify the nodes that must be stored.

\section{Space Efficient Algorithms} \label{sec:space-efficient} 

The {\sf naive} algorithm is time efficient but not space efficient.
This section presents several other algorithms.
First, if each $\deactivate$ is followed by garbage collection, then the implementation becomes space efficient but not time efficient.
Second, if $\deactivate$ is followed by garbage collection only at certain times, then the implementation becomes both space and time efficient, but only in an amortized sense.
Third, we present an algorithm that is both space and time efficient in a strict sense.
The last algorithm is somewhat sophisticated, and its correctness requires a non-obvious proof.
The implementation of all four algorithms,
  which fully specifies all the details,
  is available online~\cite{tb-impl}.

\begin{figure}[tb]\pseudocodesize 
\hbox to \hsize{%
\hfil%
\vtop{
\begin{alg}
\^  $\proc{gc}()$
\0  ${\it Seen}:=\proc{copy\_of}(\Active)$
\0  ${\it Level}:=\proc{convert\_to\_list}(\Active)$
\0  $i:=1$
\0  ~while~ $i<h$ ~and~ ${\it Level}$ is nonempty
\1    ${\it NextLevel}:=[]$
\1    ~for~ $y\in{\it Level}$
\2      $x:=\parent(y)$
\2      ~if~ $x\notin{\it Seen}$
\3        ${\it Seen}:=\{x\}\cup{\it Seen}$
\3        ${\it NextLevel}:=x\cdot{\it NextLevel}$
\1    ${\it Level}:={\it NextLevel}$
\1    $i:=i+1$
\0  ~for~ $y\in{\it Level}$
\1    $\proc{delete\_parent}(y)$
\end{alg}
\medskip
\begin{alg}
\^  $\proc{deactivate}(x)$
\0  $\Active:=\Active-\{x\}$
\0  $\proc{gc}()$
\end{alg}%
}\hfil\vtop{
\begin{alg}
\^  $\proc{delete\_parent}(y)$
\0  $x:=\parent(y)$
\0  ~if~ $x\ne{\sf nil}$
\1    $\children(x):=\children(x)-1$
\1    ~if~ $\children(x)=0$
\2      $\proc{delete\_parent}(x)$
\2      ~delete~ $x$
\2      $\mem := \mem - 1$
\1  $\parent(y):={\sf nil}$
\end{alg}
\medskip
\begin{alg}
\^  $\proc{add\_child}(x, y)$
\0  ~assert~ that $x\in\Active$
\0  $\parent(y):=x$
\0  $\children(x):=\children(x)+1$
\0  $\Active:=\Active\cup\{y\}$
\0  $\mem:=\mem+1$
\end{alg}
}\hfil}
\caption{
  The {\sf gc} algorithm.
  The tree buffer operations $\initialize$, $\expand$, and $\history$
    are those defined in \autoref{fig:alg-naive}.
}\label{fig:alg-gc}
\end{figure} 

\subsection{The Garbage Collecting Algorithm} 
\label{sub-but-not-both}

A space optimal implementation uses no more memory than needed to answer $\history$ queries.
To make this precise, let us define the \df{height} of a node~$x$ to be the shortest distance from~$x$ to an active node in the subtree of~$x$, were we to use the {\sf naive} algorithm.
Active nodes have height~$0$.
A node with no active node in its subtree has height~$\infty$.
Let $H_i$ be the set of nodes with height~$i$, and let $H_{<i}$ be the set of nodes with height less than~$i$.

The memory needed to answer $\history$ queries is~$\Omega(|H_{<h}|)$, and
  the {\sf gc} algorithm of \autoref{fig:alg-gc} achieves this bound.
On line~5 of~\gc, the list {\it Level\/} represents~$H_{i-1}$,
  and {\it Seen\/} represents~$H_{<i}$.
Thus, on line~13, the list {\it Level\/} represents~$H_{h-1}$,
  and {\it Seen\/} represents~$H_{<h}$.
The procedure {\sc delete\_parent} implements a reference counting scheme.

Let us consider a sequence of $\addchild$ and $\deactivate$ operations,
  coming after $\initialize$.
We call $\addchild$ and $\deactivate$ \df{modifying operations}.
Let $\Hs{k}_i$ be the $H_i$ corresponding to the tree obtained after $k$ modifying operations,
  and let $\sgcs{k}$ be the space used by the {\sf gc} algorithm after $k$ modifying operations.

\begin{proposition}
Consider the {\sf gc} algorithm from \autoref{fig:alg-gc}.
The memory used after $k$ modifying operations is optimal: $\sgcs{k}\in\Theta(|\Hs{k}_{<h}|)$.
The runtime used to process $k$ modifying operations is~$\Theta(k^2)$.
\end{proposition}

The space bound is obvious.
For the time bound, the following sequence exhibits the quadratic behavior:
  $\initialize(0)$,
  $\addchild(0,1)$, $\addchild(0,2)$, $\deactivate(2)$,
  $\addchild(0,3)$, $\addchild(0,4)$, $\deactivate(4)$,
  \dots

\subsection{The Amortized Algorithm} 

\begin{wrapfigure}[11]{r}[1em]{.5\textwidth} 
\vspace{-11ex}
\pseudocodesize
\hbox to \hsize{%
\hfil\vbox{%
\begin{alg}
\^  $\proc{add\_child}(x, y)$
\0  ~assert~ that $x\in\Active$
\0  $\parent(y):=x$
\0  $\children(x):=\children(x)+1$
\0  $\Active:=\Active\cup\{y\}$
\0  $\mem:=\mem+1$
\0  ~if~ $\mem = 2 \cdot \memOld$
\1    $\proc{gc}()$
\1    $\memOld := \mem$
\end{alg}}\hfil}
\caption{
  The {\sf amortized} algorithm.
  The tree buffer operations $\initialize$, $\deactivate$, $\expand$, $\history$
    are those defined in \autoref{fig:alg-naive}.
  The subroutine $\gc$ is that defined in \autoref{fig:alg-gc}.
}\label{fig:alg-amortized}
\vspace{-20ex}
\end{wrapfigure}

Our aim is to mitigate or even solve the time problem of the {\sf gc} algorithm,
  but to retain space optimality up to a constant.
One idea is to invoke the garbage collector rarely, so that the time spent in garbage collection is amortized.
To this end, we call~$\gc$ when the number of nodes in memory has doubled
  since the end of the last garbage collection.
We obtain the {\sf amortized} algorithm from \autoref{fig:alg-amortized}.
It is here that the counters $\mem$~and~$\memOld$ are finally used.

The following theorem states that the {\sf amortized} algorithm is space efficient,
  by comparing it with the {\sf gc} algorithm,
  which is space optimal.
As before, let us consider a sequence of modifying operations.
We write $\samos{k}$ for the space used by the amortized implementation after the first $k$~operations.
Call a sequence of operations \emph{extensive} if every $\deactivate(x)$ is immediately preceded by an $\addchild(x,y)$ for some~$y$.
For example,
  a sequence is extensive if it consists of an $\initialize$ operation followed by $\expand$ operations with $n>0$.

\newcommand{\stmtthmamortized}{
 Consider the {\sf amortized} algorithm in \autoref{fig:alg-amortized}.
 A sequence of $\ell$ modifying operations takes $O(\ell)$ time.
 We have $\samos{k} \in O\big( \max_{j \le k} \sgcs{j} \big)$ for all $k \le \ell$.
 If the sequence is extensive then $\samos{k} \in O\big( \sgcs{k} \big)$ for all $k \le \ell$.
}
\begin{theorem} \label{thm-amortized}
\stmtthmamortized
\end{theorem}
Loosely speaking, the theorem says that the space wasted in-between two garbage collections is bounded by the space that would be needed by the space optimal implementation at some earlier time, up to a constant.
It also says that the time used is optimal \emph{for a sequence} of operations.

\subsection{The Real-Time Algorithm} 

In general, interactive applications should not have amortized implementations.
Interactive applications include graphical user interfaces,
  but also real-time systems and
  runtime verification monitors for real-time systems.
More generally speaking,
 the environment, be it human or machine, does not accumulate patience as the time goes by.
Thus, time bounds that apply to each operation are preferable
  to bounds that apply to the sequence of operations performed so far.

The difficulty of designing a {\sf real-time} algorithm stems from the fact that whether a node is needed depends on its height,
but the heights cannot be maintained efficiently.
This is because one $\deactivate$ operation may change the heights of many nodes, possibly far away.

The key idea is to under-approximate the set of unneeded nodes;
  that is, to find a property
    that is easily computable, and only unneeded nodes have it.
To do so, we maintain three other quantities instead of heights.
The \emph{depth} of a node is its distance to the root via $\parent$ pointers,
  were we to use the {\sf naive} algorithm.
The \emph{representative} of a node is its closest ancestor whose depth is a multiple of~$h$.
The \emph{active count} of a node is the number of active nodes that have it as a representative.
Unlike height, these three quantities --- depth, representative, active count --- are easy to maintain explicitly in the data structure.
The depth only needs to be computed when the node is added to the tree.
The representative of a node is either itself or the same as the representative of its parent, depending on whether the depth is a multiple of~$h$.
Finally, when a node is deactivated (added to the tree, respectively), only one active count changes: the active count of the node's representative is decreased (increased, respectively) by~one.

The active count of a representative becomes~$0$ only if its height is at least~$h$, which means it is unneeded to answer subsequent $\history$ queries.
Thus, the set of nodes that are representatives and have an active count of~$0$
  constitutes an under-approximation of the set of unneeded nodes.
The resulting {\sf real-time} algorithm appears in \autoref{fig:alg-realtime}.

\begin{figure}[tb]\pseudocodesize 
\hbox to \hsize{%
\hfil\vtop{
\begin{alg}
\^  $\proc{initialize}(x)$
\0  $\Active:=\{x\}$
\0  $\parent(x):={\sf nil}$
\0  $\children(x):=0$
\0  $\depth(x):=0$
\0  $\rep(x):=x$
\0  $\cnt(x):=1$
\end{alg}
\medskip
\begin{alg}
\^  $\proc{process\_queue}()$
\0  ~if~ $\queue$ is nonempty
\1    $x := \deque()$
\1    $\cutParent(x)$
\1    ~delete~ $x$
\end{alg}
\medskip
\begin{alg}
\^  $\proc{deactivate}(x)$
\0  $\Active:=\Active-\{x\}$
\0  ${\it cnt}({\it rep}(x)) := {\it cnt}({\it rep}(x)) - 1$
\0  ~if~ ${\it children}(x)=0$
\1    $\proc{enque}(x)$
\0  ~if~ ${\it cnt}({\it rep}(x))=0$
\1    $\proc{cut\_parent}({\it rep}(x))$
\0    $\proc{process\_queue}()$
\end{alg}
}\hfil\vtop{
\begin{alg}
\^  $\proc{add\_child}(x, y)$
\0  ~assert~ that $x\in\Active$
\0  ~assert~ that $\cnt(y)=\children(y)=0$
\0  $\Active:=\Active\cup\{y\}$
\0  $\parent(y):=x$
\0  $\children(x):=\children(x)+1$
\0  $\depth(y) := \depth(x) + 1$
\0  ~if~ ${\it depth}(y)\equiv0 \pmod h$
\1    ${\it rep}(y) := y$
\0  ~else~
\1    ${\it rep}(y) := {\it rep}(x)$
\0  ${\it cnt}({\it rep}(y)) := {\it cnt}({\it rep}(y)) + 1$
\0    $\proc{process\_queue}()$
\end{alg}
\medskip
\begin{alg}
\^  $\proc{cut\_parent}(y)$
\0  $x:=\parent(y)$
\0  ~if~ $x\ne{\sf nil}$
\1    ${\it children}(x) := {\it children}(x)-1$
\1    ~if~ ${\it children}(x)=0$ ~and~ $x \not\in\Active$
\2      $\enque(x)$
\0  $\parent(y) := {\sf nil}$
\end{alg}
}\hfil}
\caption{
  The {\sf real-time} algorithm.
  The tree buffer operations $\expand$ and $\history$
    are those defined in \autoref{fig:alg-naive}.
  The $\enque$ and $\deque$ operations are the standard operations of a queue data structure.
}\label{fig:alg-realtime}
\end{figure}

As $\textsc{delete\_parent}$ did in the {\sf gc} algorithm,
  the function $\deactivate$ implements a reference counting scheme,
  using $\children$ as the counter.
Unlike the {\sf gc} algorithm, the node is not deleted immediately,
  but \emph{scheduled for deletion}, by being placed in a queue.
This queue is processed whenever the user calls $\addchild$ or $\deactivate$.
When the queue is processed, by $\textsc{process\_queue}$,
  one node is deleted from memory,
  and perhaps its parent is scheduled for deletion.

The proof of the following
  theorem\iftechrep{, provided in \hyperref[sub-proof-thm-main]{Appendix~\ref*{sub-proof-thm-main}},}{~\cite{gk-cav15-tr}} is subtle.
Similarly as before, we write $\srts{k}$ for the space that the {\sf real-time} algorithm has allocated and not deleted after $k$ operations.

\newcommand{\stmtthmmain}{
Consider the {\sf real-time} algorithm from \autoref{fig:alg-realtime},
  and a sequence of $\ell$ modifying operations.
Every operation takes $O(1)$ time.
We have $\srts{k} \in O\big( \max_{j \le k} \sgcs{j} \big)$ for all $k \le \ell$.
If the sequence is extensive then $\srts{k} \in O\big( \sgcs{k} \big)$ for all $k \le \ell$.
}
\begin{theorem} \label{thm-main}
\stmtthmmain
\end{theorem}

\section{Monitoring}\label{sec:monitoring} 

Consider a nondeterministic automaton
$\A = (Q,E,q_0,F,\delta_i,\delta_r)$, where
  $Q$~is a set of states,
  $E$ is the alphabet of events,
  $q_0\in Q$ is the initial state,
  $F\subseteq Q$ contains the accepting states, and
  $\delta_i,\delta_r\subseteq Q\times E\times Q$ are, respectively, the irrelevant and the relevant transitions.
We aim to construct a monitor that reads a stream of events and reports an error trace when an accepting state has been reached.
Since $\A$ is in general nondeterministic and there are both irrelevant and relevant transitions, building an efficient monitor for~$\A$ is not straightforward.
We have sketched in the introduction how to use a tree buffer for such a monitor.
The algorithm in \autoref{fig:alg-nfa} makes this precise.

The main invariants (line~4) are the following:
\begin{itemize}
\item
If the pair $(q,{\it node})$ is in the list $\it now$,
  then $\history(node)$ would return the last $\le h$ relevant transitions
  of some computation $q_0 \mathrel{\stackrel{w}{\to}\!\!{}^*} q$ of~$\A$,
  where $w$~is the stream read so far.
\item
If there is a computation $q_0 \mathrel{\stackrel{w}{\to}\!\!{}^*} q$ of~$\A$,
  then, after reading~$w$, a pair $(q, {\it node})$ is in the list $\it now$, for some $\it node$.
\end{itemize}
A node~$x$
  is created and added to the tree buffer when a relevant transition is taken (lines~10--11).
The node~$x$ is deactivated (line~19)
  when and only when it is about to be removed from the list $\it now$ (line~20),
  since neither $\addchild(x,\cdot)$ nor $\history(x)$ can be invoked later.

In the following subsections we give two applications for this monitor.
The $\it location$, which accompanies events (lines 5~and~10),
  is application dependent.
For regular expression searching, the $\it location$ is an index in a string;
for runtime verification, the $\it location$ is a position in the program text.

\begin{figure}[tb]\pseudocodesize 
\begin{alg}
\^  $\proc{monitor}()$
\0  ${\it root\_node} := \textsc{make\_node}(\rtran{\ }\!q_0,\;{\sf nil})$
\0  $\initialize({\it root\_node})$
\0  ${\it now}, {\it nxt} := [(q_0, {\it root\_node})], []$
\0  ~forever~
\1    $a, {\it location} := \proc{get\_next\_event\_and\_location}()$
\1    ~for~ each $(q, {\it parent})$ in the list $\it now$
\2      ~for~ each $a$-labeled transition $t = (q \stackrel{a}{\to} q') \in \delta_i \uplus \delta_r$
\3        ~if~ $\lnot{\it in\_nxt}(q')$
\4          ~if~ $t \in \delta_r$
\5            ${\it child} := \proc{make\_node}(t, {\it location})$
\5            $\addchild({\it parent}, {\it child})$
\4          ~if~ $t \in \delta_i$
\5            ${\it child} := {\it parent}$
\4          append $(q', {\it child})$ to $\it nxt$
\4          ${\it in\_nxt}(q'), {\it in\_nxt}({\it child}) := {\sf true}, {\sf true}$
\4          ~if~ $q' \in F$
\5            $\proc{report\_error}(\history({\it child}))$
\1    ~for~ each $(q, {\it node})$ in the list $\it now$
\2      ~if~ $\lnot{\it in\_nxt}({\it node})$ ~then~ $\deactivate({\it node})$
\1    ${\it now}, {\it nxt} := {\it nxt}, []$
\1    ~for~ each $(q, {\it node})$ in the list $\it now$
\2      ${\it in\_nxt}(q), {\it in\_nxt}({\it node}) := {\sf false}, {\sf false}$
\end{alg}
\caption{
  A monitor for the automaton $\A = (Q,E,q_0,F,\delta_i,\delta_r)$.
  The monitor reports error traces by using a tree buffer.
}
\label{fig:alg-nfa}
\end{figure}

\subsection{Regular-Expression Searching} \label{sub:regex} 

We show that regular-expression searching with \emph{capturing groups}
can be implemented by constructing an automaton with irrelevant and relevant transitions,
 and then running the monitor from \autoref{fig:alg-nfa}.
Suppose we want to search Wikipedia for famous people with reduplicated names,
  like `Ford Madox Ford'. 
One approach is to use the following (Python) regular expression:
\begin{align}\label{eq:ex-regex}
  \texttt{Ford(\textvisiblespace[A-Z][a-z]*)\{$m$,$n$\}\textvisiblespace Ford}
\end{align}
This expression matches names starting and ending with `Ford', and with at least $m$ and at most~$n$ middle names in-between.
The parentheses indicate so-called \emph{capturing groups}:
The regular-expression engine is asked to remember (and possibly later output) the position in the text
 where the group was matched.
We can implement this as follows.
First, we compile the regular expression with capturing groups into an automaton
  with relevant and irrelevant transitions.
Which transitions are relevant could be determined automatically using the capturing groups,
 or the user could specify it using a special-purpose extension of the syntax of regular expressions.
Whenever the automaton takes a relevant transition, the position in the text should be remembered.
Then we run the monitor from \autoref{fig:alg-nfa} on this automaton.
In this way we can output the last $h$ matches of capturing groups.
In contrast, standard regular-expression engines would report only the last occurrence of each match.
In the example expression~\eqref{eq:ex-regex}, they would report only the last of Ford's middle names.
One would have to unroll the expression $n$ times in order to make a standard engine report them all.

For the regular expression~\eqref{eq:ex-regex}, we remark that any equivalent \emph{deterministic} automaton has $\Omega(2^m)$ states, so nondeterminism is essential for feasibility%
\footnote{
  We use a large value for $m$ when we want to find people with reduplicated names
    that are long.
  By searching Wikipedia with large values for $m$ we found, for example,
    `Jos\'e Mar\'ia del Carmen Francisco Manuel Joaqu\'in \emph{Pedro} Juan Andr\'es Avelino Cayetano Venancio Francisco de Paula Gonzaga Javier Ram\'on Blas Tadeo Vicente Sebasti\'an Rafael Melchior Gaspar Baltasar Luis \emph{Pedro} de Alc\'antara Buenaventura Diego Andr\'es Apostol Isidro'
    (a Spanish don).
}.

\subsection{Runtime Verification} \label{sub:rv} 

For runtime verification we use the monitor from \autoref{fig:alg-nfa} as well,
 in the way we sketched in the introduction.
Clearly, for real-time runtime verification the {\sf real-time} tree buffer algorithm needs to be used.

We have not yet emphasized one feature of our monitor,
 which is essential for runtime verification:
The automaton $\A = (Q,E,q_0,F,\delta_i,\delta_r)$ may have an infinite set $Q$ of states,
 and it may deal with infinite event alphabets~$E$.
Note that we did not require any finiteness of the automaton for our monitor.
We can implement the monitor from \autoref{fig:alg-nfa},
 as long as we have a finite \emph{description} of~$\A$,
 which allows us to loop over transitions (line~7) and to store individual states and events.
One can view this as constructing the (infinite) automaton on the fly.
For instance, the event alphabet could be
$
 E = \Sigma \times {\sf Value}
$, 
 where $\Sigma=\{{\it iter}, {\it hasNext}, {\it next}, {\it other}\}$ and
 ${\sf Value}$ is the set of all program values, which includes integers, booleans, object references, and so on.
There are various works on automata over infinite alphabets and with infinitely many states.
In those works, infinite (-state or -alphabet) automata are usually called \emph{configuration graphs},
 whereas the word \emph{automaton} refers to a finite description of a configuration graph.
In contrast to the rest of the paper, we use that terminology in the rest of this paragraph.
Often there exists an explicitly defined translation of an automaton to a configuration graph
  (for example, for
  register automata~\cite{dblp:conf/focs/kaminskif90},
  class memory automata~\cite{dblp:journals/tcs/bjorklunds10}, and
  history register automata~\cite{dblp:conf/fossacs/tzevelekosg13}).
Even when the semantics are not given in terms of a configuration graph,
  it is often easy to devise a natural translation.
For example, the configuration graph in \autoref{fig:hasnext-unfolded}
 is obtained from the automaton of \autoref{fig-intro-NFAs}a using an obvious translation
 that would also apply
  in the case of data automata~\cite{dblp:conf/lics/bojanczykmssd06}
  and in the case of slicing~\cite{dblp:journals/corr/abs-1112-5761}.

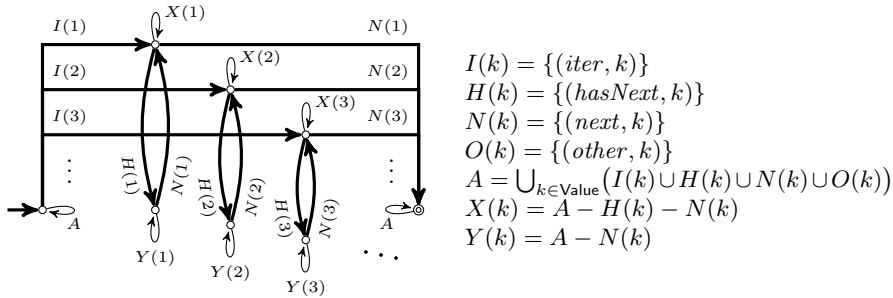
\begin{figure}[tb] 
\begin{tikzpicture}[auto]
\tikzset{state/.style={draw,circle,inner sep=1pt}}
\node[initial,state] (start) at (0,0) {};
\node[accepting,state] (error) at (5,0) {};
\draw[irrelevant transition]
  (start) edge[loop right] node[below] {$\scriptstyle A$} ()
  (error) edge[loop left] node[below] {$\scriptstyle A$} ()
;

\def\middlepart#1{
  \node[state] (#1-invalid) at ($(.5+#1,2.8-0.6*#1)$) {};
  \node[state] (#1-valid) at ($(.5+#1,0.2-0.2*#1)$) {};
  \draw[relevant transition]
    (start) |- node[anchor=south west]{$\scriptstyle I(#1)$} (#1-invalid);
  \draw[relevant transition]
    (#1-invalid) -| node[anchor=south east]{$\scriptstyle N(#1)$} (error);
  \draw[relevant transition]
    (#1-invalid) edge[bend right=15]
      node[anchor=north,sloped,pos=.8]{$\scriptstyle H(#1)$} (#1-valid)
    (#1-valid) edge[bend right=15]
      node[anchor=north,sloped,pos=.2]{$\scriptstyle N(#1)$} (#1-invalid)
  ;
  \draw[irrelevant transition]
    (#1-invalid) edge[loop above] node[right] {$\scriptstyle X(#1)$} ()
    (#1-valid) edge[loop below] node[below] {$\scriptstyle Y(#1)$} ()
  ;
}
\middlepart{1}
\middlepart{2}
\middlepart{3}
\node[anchor=north west] at (0.2,1) {$\vdots$};
\node[anchor=north east] at (4.8,1) {$\vdots$};
\foreach \p in {-1,0,1} {
  \fill ($(4.5+0.2*\p,-0.6-0.04*\p)$) circle (.7pt);
}
\node[anchor=north west] at (5.5,2.2) {\vbox{\hsize=5.7cm
  $I(k) = \{({\it iter}, k)\}$\\
  $H(k) = \{({\it hasNext}, k)\}$\\
  $N(k) = \{({\it next}, k)\}$\\
  $O(k) = \{({\it other}, k)\}$\\
  $A =\bigcup_{k\in{\sf Value}} \bigl(I(k)\cup H(k) \cup N(k) \cup O(k)\bigr)$\\
  $X(k)=A-H(k)-N(k)$ \\
  $Y(k)=A-N(k)$
}};
\end{tikzpicture}
\caption{
  The configuration graph of \autoref{fig-intro-NFAs}a.
  The arcs are labeled by \emph{sets} of events,
    meaning that there is one transition for each event in the set.
  The picture shows only three values from ${\sf Value}=\{1,2,3,\ldots\,\}$
}
\label{fig:hasnext-unfolded}
\end{figure}

\section{Experiments}\label{sec:experiments} 

This section complements the asymptotic results of \autoref{sec:space-efficient}
  with experimental results from three data sets.
\iftechrep{
  The implementation, datasets, and experimental logs are available online~\cite{tb-impl}.
}{}

\subsection{Datasets} 

\begin{enumerate}
\item
The first dataset is a sequence of $n = 10^7$ operations
  that simulate a sequence of linear buffer operations.
That is, we called the tree buffer as follows:
$
  \initialize(0);\;\;
  \expand(0,\{1\});\;
  \ldots\,;\;
  \expand(n-1,\{n\})
$.
\item
We produced (manually) the automaton in \autoref{fig:aa-automaton} from the regular expression
  `{\tt .*a(\textvisiblespace *[\^{} ])\{8\}\textvisiblespace *a}',
  and ran the monitor from \autoref{sec:monitoring} on the text of Wikipedia.
This dataset contains $7\cdot10^8$ tree buffer operations.
\begin{figure}[htb]\centering 
\begin{tikzpicture}[auto]
\tikzset{state/.style={draw,circle,inner sep=1pt}}
\foreach \x in {1,2,...,9} {
  \node[state] (\x) at (\x,0) {$\x$};
}
\node[state,initial] (0) at (0,0) {$0$};
\node[state,accepting] (10) at (10,0) {\phantom{$0$}};
\path[relevant transition]
  \foreach \x in {1,2,...,9} {
    (\x) edge[loop below] node {\textvisiblespace} ()
  }
  (0) edge node {{\it a}} (1)
  (9) edge node {{\it a}} (10)
;
\path[irrelevant transition]
  \foreach \x [evaluate=\x as \sx using int(\x+1)] in {1,2,...,8} {
    (\x) edge node {[\^{}\textvisiblespace]} (\sx)
  }
  (0) edge[loop below] node{*} ()
;
\end{tikzpicture}
\caption{
  A nondeterministic automaton without a small, deterministic equivalent:
  It finds substrings that contain $10$ non-space characters,
    the first and last of which are~`a'.
  The structure of the automaton is similar to the one
  corresponding to the regular expression from \autoref{sub:regex}.
}
\label{fig:aa-automaton}
\end{figure}
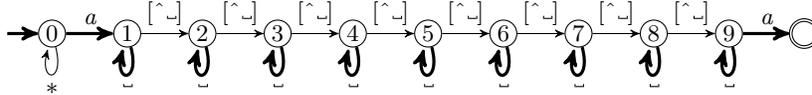
\item
We ran the monitor from \autoref{sec:monitoring} on infinite automata alongside the DaCapo test suite.
The property we monitored
  was specified using a TOPL automaton~\cite{dblp:conf/tacas/grigoredpt13},
  and it was essentially the one in \autoref{fig-intro-NFAs}a:
it is an error if there is a $\textsc{next}$ without a preceding $\textsc{hasNext}$ that returned $\sf true$.
We used the projects
  avrora (simulator of a grid of microcontrollers),
  eclipse (development environment),
  fop (XSL to PDF converter),
  h2 (in memory database),
  luindex (text indexer),
  lusearch (text search engine),
  pmd (simple code analyzer),
  sunflow (ray tracer),
  tomcat (servlet server),
  and xalan (XML to HTML converter)
from version~9.12 of the DaCapo test suite~\cite{dblp:conf/oopsla/blackburnghkmbdffghhjlmpsvdw06}.
This dataset contains $8\cdot10^7$ tree buffer operations.
\end{enumerate}

\subsection{Empirical Results} 

\begin{figure}[p] 
\def\lf#1#2{\begin{minipage}{2in}\centering\includegraphics{plots/#1}\par#2\end{minipage}}
\hbox to\hsize{\hss%
\begin{minipage}{6.1in}
\lf{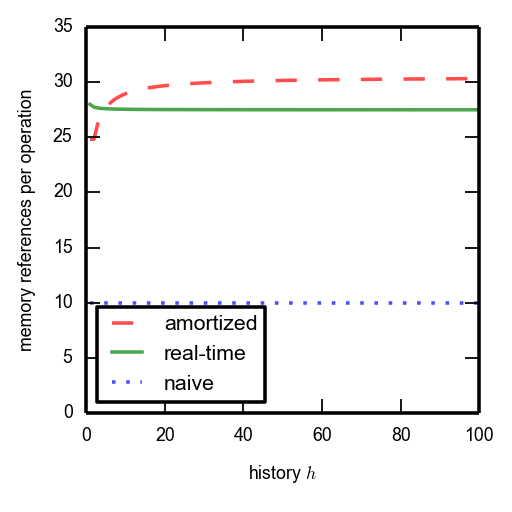}{(a)~as linear buffers}
\lf{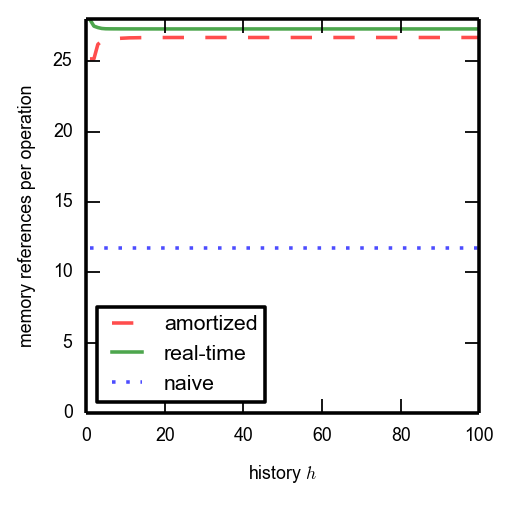}{(b)~regular expression searching}
\lf{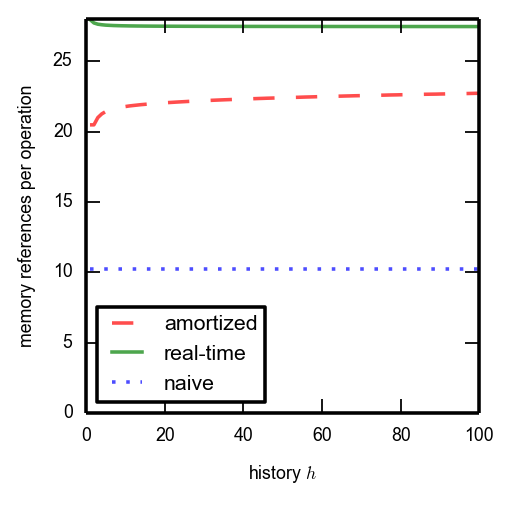}{(c)~runtime verification}
\end{minipage}\hss}
\caption{
  The average number of memory references per tree buffer operation.
}\label{fig:exp-runtime-nogc}
\hbox to\hsize{\hss%
\begin{minipage}{6.1in}
\lf{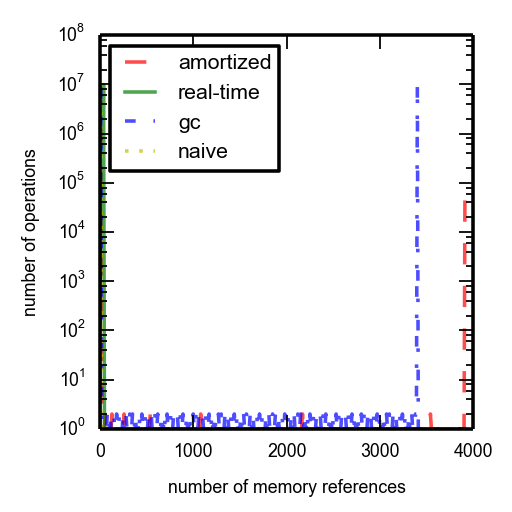}{(a)~as linear buffers}
\lf{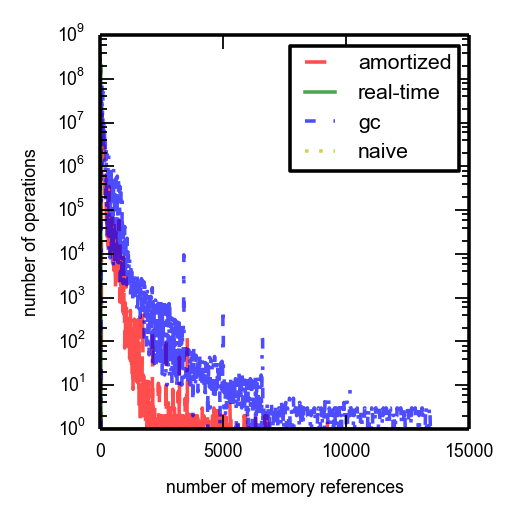}{(b)~regular expression searching}
\lf{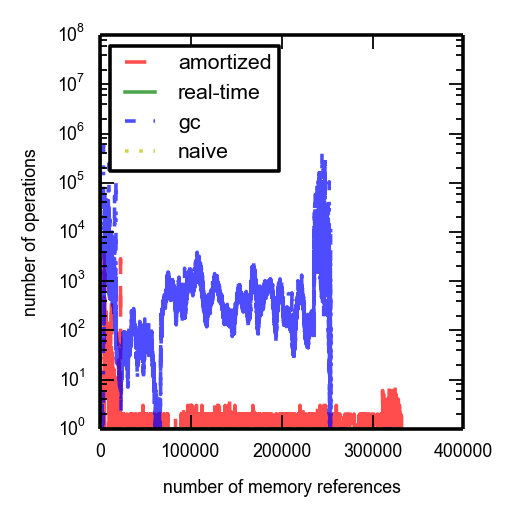}{(c)~runtime verification}
\end{minipage}\hss}
\caption{
  Histogram for the number of memory references per operation, for $h=100$.
}\label{fig:exp-runtime-variability}
\hbox to\hsize{\hss%
\begin{minipage}{6.1in}
\lf{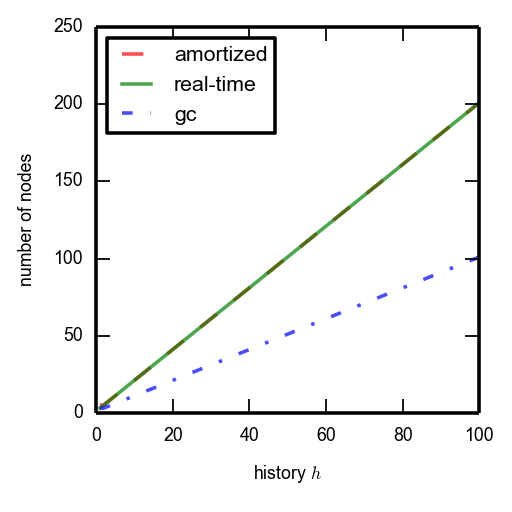}{(a)~as linear buffers}
\lf{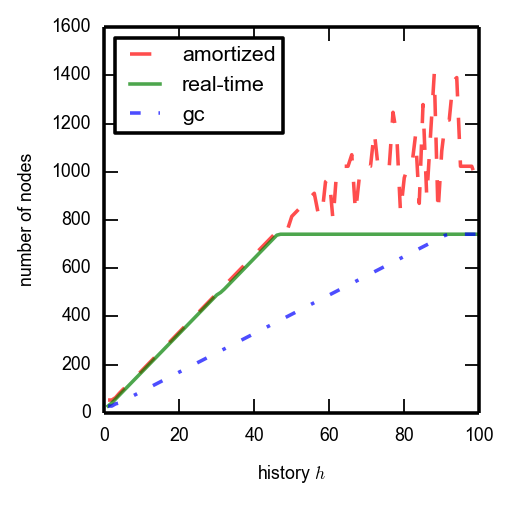}{(b)~regular expression searching}
\lf{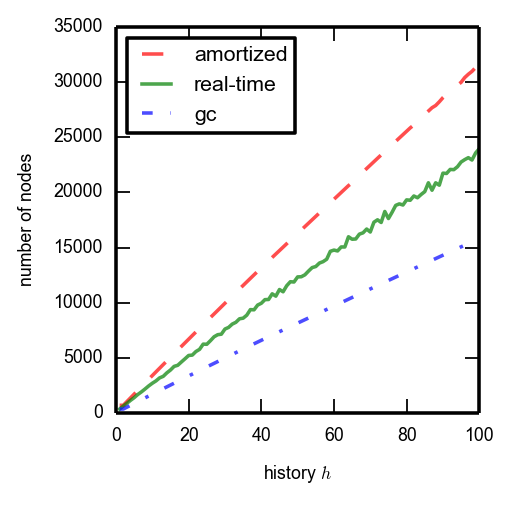}{(c)~runtime verification}
\end{minipage}\hss}
\caption{
  How much space is necessary.
}\label{fig:exp-space}
\end{figure}

We measure space and time in a way that is machine independent.
For space, there is a natural measure: the number of nodes in memory.
For time, it is less clear what the best measure is:
We follow Knuth~\cite{dblp:books/daglib/0071477}, and count memory references.

\paragraph{Runtime versus History.}

\autoref{fig:exp-runtime-nogc} gives
  the \emph{average} number of memory references per operation.
We observe that this number does not depend on~$h$,
  except for very small values of~$h$,
  thus validating the asymptotic results about time from \autoref{sec:space-efficient}.
\iftechrep{
  \autoref{fig:exp-runtime} in \autoref{sec:additionalmaterial} confirms
  that the {\sf gc} algorithm is much slower than the others.
}{}

\paragraph{Runtime Variability.}

\autoref{fig:exp-runtime-variability} shows that
  for the {\sf amortized} and {\sf gc} algorithms
  there exist operations that take a long time.
In contrast, the plots for the {\sf naive} and the {\sf real-time} algorithms
  are almost invisible
  because they are completely concentrated on the left side
    of \autoref{fig:exp-runtime-variability}.

\paragraph{Memory versus History.}

In \autoref{fig:exp-space},
  we notice that the memory usage of the {\sf amortized} and the {\sf real-time}
  algorithms is within a factor of~$2$ of the memory usage of the {\sf gc} algorithm,
  thus validating the asymptotic results about space from \autoref{sec:space-efficient}.
The {\sf naive} algorithm is excluded from \autoref{fig:exp-space}
  because its memory usage is much bigger than that of the other algorithms.

\section{Conclusions, Related Work, and Future Work} 

We have designed \emph{tree buffers}, a data structure that generalizes linear buffers.
A tree buffer consumes a stream of events
  each of which declares its parent to be one of the preceding events.
Tree buffers can answer queries that ask for the $h$~ancestors
  of a given event.
Implementing tree buffers with good performance is not easy.
We have explored the design space by developing four possible algorithms
  ({\sf naive}, {\sf gc}, {\sf amortized}, {\sf real-time}).
Two of those are straightforward:
{\sf naive} is time optimal, and {\sf gc} is space optimal.
The other two algorithms are time and space optimal at the same time:
  {\sf amortized} is simpler but not suitable for real-time use,
  and {\sf real-time} is more involved but suitable for real-time use.
Proving the {\sf amortized} and the {\sf real-time} algorithms correct requires some care.
We have validated our algorithms on data sets from three different application~areas.

Algorithms that process their input in a gradual manner have been studied under the names of
  online algorithms, dynamic data structures, and, more recently, streaming algorithms.
These algorithms address different problems than tree buffers.
For example, streaming algorithms~\cite{streams1,streams2} fall into two classes:
  those that process numeric streams, and those that process graph streams.
Graph streaming algorithms are concerned with problems such as:
  `Are vertices $u$~and~$v$ connected in the graph described so far?'
One of the basic tools used for answering such questions are link-cut trees~\cite{link-cut}.
Yet, like all the existing graph streaming algorithms,
  link-cut trees do not give more weight to the recent parts of the tree,
  in the way tree buffers do.
Such a preference for recent data has been studied only in the context of numeric streams.
For example, the following problem has been studied:
  `Which movie is most popular \emph{currently}?'
\cite[Section 4.7]{streams2}

The closest relatives of tree buffers remain the simple and ubiquitous linear buffers.
Since tree buffers extend linear buffers naturally,
 it is easy to imagine a wide array of applications.
We have discussed an engine for regular expression searching as one example.
The main motivation of our research is to enhance
 \emph{runtime verification monitors} with the ability to provide error traces,
 fulfilling real-time constraints if needed,
 and covering general nondeterministic automata specifications.
We have described this application in detail.

Several automata models that are used in runtime verification,
  including the TOPL automata used in our implementation,
  are nondeterministic
  \cite{%
    dblp:conf/tacas/grigoredpt13,%
    dblp:conf/isola/havelund14,%
    dblp:journals/corr/abs-1112-5761%
  },
  which led us to a tree data structure that can track such automata.
Some automata models are even more general,
 such as
    quantified event automata~\cite{dblp:conf/fm/barringerfhrr12} and
    alternating automata~\cite{dblp:journals/fmsd/finkbeiners04}.
The construction of error-trace providing monitors for such automata is an intriguing challenge that seems to raise further fundamental algorithmic questions.

\paragraph{Acknowledgements.} 

Grigore is supported by EPSRC Programme Grant Resource Reasoning (EP/H008373/2).
Kiefer is supported by a Royal Society University Research Fellowship.
We thank the reviewers for their comments.
We thank Rasmus Lerchedahl Petersen
  for his contribution to the implementation of an early version of the {\sf amortized} algorithm
  in the runtime verifier TOPL\null.

{\raggedright
\bibliographystyle{plain}
\bibliography{db}}

\iftechrep{
\newpage\appendix
\section{Additional Graphs}\label{sec:additionalmaterial} 

\begin{figure}[h] 
\def\lf#1#2{\begin{minipage}{2in}\centering\includegraphics{plots/#1}\par#2\end{minipage}}
\hbox to\hsize{\hss%
\begin{minipage}{6.1in}
\lf{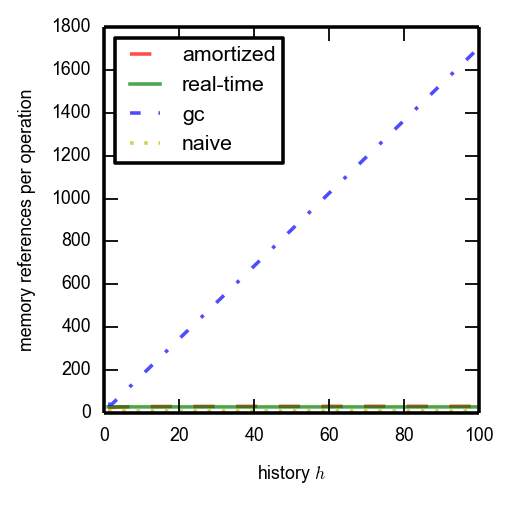}{(a)~as linear buffers}
\lf{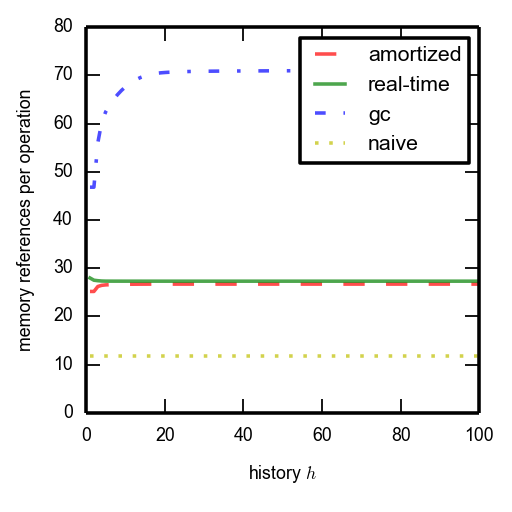}{(b)~regular expression searching}
\lf{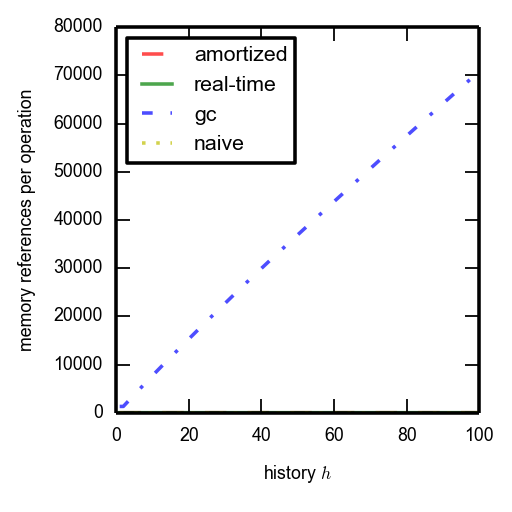}{(c)~runtime verification}
\end{minipage}\hss}
\caption{
  The average number of memory references per tree buffer operation.
  Unlike \autoref{fig:exp-runtime-nogc}, these plots include the {\sf gc} algorithm.
}\label{fig:exp-runtime}
\end{figure}
%
\section{Proofs}\label{sec:proofs} 

All results talk about sequences of modifying operations,
  but this is without loss of generality:
(1)~any call to $\history$ takes $\Theta(1)$ space and $O(h)$ time in all algorithms;
(2)~any call to $\expand(x,\{y_1,\ldots,y_n\})$ is equivalent to the segment of operations
\begin{align*}
  \addchild(x,y_1);\; \ldots\;;\; \addchild(x,y_n);\; \deactivate(x)
\end{align*}
Given these observations,
  we can use the results from below to deduce the space and time usage
  of any sequence of operations.

The following lemma about extensive sequences will be used in the proofs of Theorems \ref{thm-amortized}~and~\ref{thm-main}.
\begin{lemma} \label{lem-extensive-nondecr}
Consider an extensive sequence of $\ell$~operations.
Let $n \ge 1$.
Then for all $i,j$ with $0 \le i \le j \le \ell$ we have
$|\Hs{i}_{<n}| - 1 \le |\Hs{j}_{<n}|$.
\end{lemma}

\begin{proof}
We first establish these two facts:
\begin{align}
  |\Hs{i}_{<n}| - 1 &\le |\Hs{i+1}_{<n}| &&\text{for $0 \le i < \ell$}
    \label{eq:decrease-1}\\
  |\Hs{i}_{<n}| &\le |\Hs{i+2}_{<n}| &&\text{for $0 \le i < \ell - 1$}
    \label{eq:decrease-2}
\end{align}

For \eqref{eq:decrease-1}, we do a case analysis on the $(i+1)$th operation.
The interesting case is that in which the $(i+1)$th operation is a $\deactivate(x)$,
  for some~$x$.
Because the sequence is extensive, the $i$th operation must be $\addchild(x,y)$, for some~$y$.
Consider now an arbitrary node $z\in\Hs{i}_{<n}$.
By the definition of $\Hs{i}_{<n}$,
  there must exist an active node~$u$ such that $z=\parent^k(u)$, for some $k<n$.
If $u\ne x$, then $u$ remains active after the $\deactivate(x)$ operation,
  and hence $z\in\Hs{i+1}_{<n}$.
If $u=x$, then $z=\parent^{k+1}(y)$.
In this case, if $k+1<n$, then again $z\in\Hs{i+1}_{<n}$.
Thus, there is at most one element of $\Hs{i}_{<n}$ that might not belong to $\Hs{i+1}_{<n}$,
  namely $\parent^{n-1}(x)$.
We proved~\eqref{eq:decrease-1}.

For \eqref{eq:decrease-2}, note that in an extensive sequence
  at most one of the $(i+1)$th and $(i+2)$th modifying operations is a $\deactivate$.
Given~\eqref{eq:decrease-1} and given that $\addchild$ increases by~$1$
  the number of active nodes, \eqref{eq:decrease-2}~follows.

Now, take $i$~and~$j$ such that $i\le j$.
By repeated application of~\eqref{eq:decrease-2}
  we know that $|\Hs{i}_{<n}|\le|\Hs{i+2p}_{<n}|$,
  for all~$p$ such that $0\le i+2p\le\ell$.
In particular, either $|\Hs{i}_{<n}|\le|\Hs{j}_{<n}|$ or $|\Hs{i}_{<n}|\le|\Hs{j-1}_{<n}|$.
In the first case we are done;
in the second case we find the desired result by using~\eqref{eq:decrease-1}.
\qed
\end{proof}

\subsection{Proof of \autoref{thm-amortized}} 

\begin{qtheorem}{\ref{thm-amortized}}
\stmtthmamortized
\end{qtheorem}

A \emph{garbage collection cycle}
  is a segment~$\sigma$ of some sequence of modifying operations such that
\begin{itemize}
\item
  the first operation of $\sigma$ follows immediately
    after an operation that triggered a garbage collection, or
    after $\initialize$; and
\item
  the operations of~$\sigma$ do not trigger a garbage collection,
    except possibly the last operation.
\end{itemize}
We begin by proving the following lemma.
\begin{lemma} \label{lem-amortized-segment}
There exists a constant~$c$ such that the runtime of any garbage collection cycle~$\sigma$
  is at most $c \cdot k$,
  where $k$~is the length of~$\sigma$.
\end{lemma}
\begin{proof}
Recall the implementation from \autoref{fig:alg-amortized}.
Each modifying operation that does not trigger the garbage collector takes $\le c_1$ time,
  for some constant~$c_1$.
Thus, if $\sigma$ does not trigger the garbage collector then its runtime is $\le c_1\cdot k$.
It remains to check the case in which the last operation of~$\sigma$
  does trigger the garbage collector.

The time spent in the garbage collector is $\le c_2\cdot\mem$, for some constant~$c_2$.
In order to find an upper bound for~$\mem$, we make two observations:
\begin{itemize}
\item when the garbage collector is triggered, $\mem = 2\cdot\memOld$, and
\item
  the number $\mem-\memOld$ of nodes added to the tree
  is the number of $\addchild$ operations in~$\sigma$
  which in turn is at most~$k$
\end{itemize}
Combining these two observations we get that $\mem\le 2\cdot k$.

We can now compute a bound for the total runtime of~$\sigma$:
\begin{align*}
 c_1\cdot k + c_2 \cdot\mem \le  c_1 \cdot k + c_2 \cdot (2 \cdot k) = (c_1 + 2 c_2)\cdot k
\end{align*}
Thus, $c:=c_1+2c_2$ has the required property.
\qed
\end{proof}


Now we prove \autoref{thm-amortized}.
\begin{proof}[of \autoref{thm-amortized}]
Consider any sequence~$\sigma$ of $\ell$ modifying operations.
First we prove the statement on time complexity.
The sequence~$\sigma$ can be decomposed into garbage collection cycles.
Applying \autoref{lem-amortized-segment} to each garbage collection cycle,
  and summing up the runtimes,
  we obtain that $\sigma$ takes at most $c \cdot \ell$ time.
This is $O(\ell)$ time.

Next we prove the statements on space complexity.
Pick an arbitrary $k \le \ell$.
Let $k_0 \ge 0$ be the largest number so that $k_0 \le k$ and either $k_0 = 0$ or the $k_0$th operation triggered a garbage collection.
For any $i \ge 0$ write $\mems{i}$ for the value of~$\mem$ after the $i$th operation.
The garbage collection ensures $\mems{k_0} = |\Hs{k_0}_{<h}|$.
Further, the implementation of~$\addchild$ ensures $\mems{k} \le 2 \cdot \mems{k_0}$,
 and so $\mems{k} \le 2 \cdot |\Hs{k_0}_{<h}|$.
For all~$i$ we have $\samos{i} \in \Theta(\mems{i})$ and $\sgcs{i} \in \Theta(|\Hs{i}_{<h}|)$.
It follows $\samos{k} \in O\big( \sgcs{k_0} \big)$ and hence $\samos{k} \in O\big( \max_{j \le k} \sgcs{j} \big)$, which is the first of the two statements on space complexity.
For the second one, assume that $\sigma$ is extensive.
By \autoref{lem-extensive-nondecr} we have $|\Hs{k}_{<h}| \ge |\Hs{k_0}_{<h}| - 1$,
so \[ \mems{k} \le 2 \cdot |\Hs{k_0}_{<h}| \le 2 \cdot \left(|\Hs{k}_{<h}|+1\right)\;,\]
and hence $\samos{k} \in O\big( \sgcs{k} \big)$.
\qed
\end{proof}

\subsection{Proof of \autoref{thm-main}} \label{sub-proof-thm-main} 

In the following, consider the tree obtained in the reference implementation after a fixed sequence of modifying operations.
By~$\Nodes$ we denote the set of nodes of the tree.
The following lemma states a monotonicity property of~$|H_i|$:

\begin{lemma} \label{lem-heights-descending}
We have $|H_{i}| \ge |H_{i+1}|$ for all $i \ge 0$.
As a consequence, we have $|H_{<2 h}| \le 2 |H_{<h}|$.
\end{lemma}
\begin{proof}
Denote by $\parent : \Nodes \to \Nodes$ the partial function that assigns to a node its parent; $\parent(x)$ is undefined for the root~$x$.
Extend~$\parent$ to $\parent : 2^\Nodes \to 2^\Nodes$ in the standard way.
Then we have $H_{i+1} \subseteq \parent(H_i)$ and $|H_i| \ge |\parent(H_i)|$.
The statement follows.
\qed
\end{proof}

Let the \df{level} of node~$x$, denoted by $\level(x)$, be $\lfloor\depth(x)/h\rfloor$.
A node~$x$ is called \emph{recent} if there exists an active node $y$ in the subtree of~$x$
  such that $\level(x)\ge\level(y)-1$.
Let $R$ denote the set of recent nodes.
\begin{lemma} \label{lem-recent-are-not-high}
We have $R \subseteq H_{< 2 h}$.
\end{lemma}
\begin{proof}
We pick an arbitrary $x\in R$, and show that $x\in H_{<2h}$.

Because $x$~is recent, there exist an active node~$y$ and an integer~$k\ge0$
  such that $\level(x)\ge\level(y)-1$ and $x=\parent^k(y)$.
Thus,
\begin{align*}
\left\lfloor\frac{\depth(x)}{h}\right\rfloor
  \ge \left\lfloor\frac{\depth(y)}{h}\right\rfloor-1
  = \left\lfloor\frac{\depth(x)+k-h}{h}\right\rfloor
\end{align*}
In general, if $\lfloor a/h\rfloor\ge\lfloor b/h\rfloor$ then $b-a<h$.
In our case, $k-h<h$, so $k<2h$.
In other words,
  if $y$ is a witness for $x\in R$,
  then $y$~is also a witness for $x\in H_{<2h}$.
\qed
\end{proof}


A node~$x$ is said to be a \df{fringe} node when $\depth(x)\equiv0 \pmod h$ and $\cnt(x)=0$.
A node~$x$ is said to be a \df{doomed} node
  when it is inactive and each of its children is either a fringe node or a doomed node.
Let $D$ denote the set of doomed nodes.
It is easy to check that the {\sf real-time} algorithm
  schedules for deletion (and then deletes) only doomed nodes.

\begin{lemma}\label{lem:doomed-or-recent} 
Every node is either doomed or recent: $\Nodes=R\uplus D$.
\end{lemma}
\begin{proof}
We prove first that a node that is not doomed must be recent;
we will later prove that a recent node must be not doomed.

Let $x$ be a node that is not doomed.
If there exists an active node~$y$ in the subtree of~$x$ such that $\level(x)=\level(y)$,
  then $x$~is recent.
Thus, for what follows, assume that no such node~$y$ exists.
In this case, we will prove by induction on $k:=h-\bigl(\depth(x)\bmod h\bigr)$
  that there exists a node~$z$ in the subtree of~$x$ such that $\level(x)=\level(z)-1$,
  and hence $x$~is, again, recent.
Note that $1\le k\le h$.

The base case is $k=1$.
By the definition of doomed,
  $x$~is active, or it has a child~$u$ that is not doomed and not fringe.
If $x$~were active, then we could take $y:=x$; so $x$ must be inactive.
Because $k=1$, it must be that $\depth(u)\equiv0\pmod h$.
Since $u$~is not fringe, it must be that $\cnt(u)>0$.
Hence, there exists an active node~$z$ and an integer $0\le l< h$ such that $u=\parent^{l}(z)$.
We have that $\level(x)=\level(u)-1=\level(z)-1$,
  and so $z$~has the desired properties.

For the induction step case, pick an arbitrary~$k$ such that $1<k\le h$.
As above, $x$~must be inactive, and must have a child~$u$ that is not doomed and not fringe.
In addition, $\level(x)=\level(u)$, because of the limits on~$k$.
By the induction hypothesis, there exists an active node~$z$ in the subtree of~$u$
  such that $\level(u)=\level(z)-1$.
This node~$z$ is also in the subtree of~$x$, and indeed $\level(x)=\level(z)-1$.

We conclude that if a node is not doomed then it is recent.

For the other direction, let $x$ be a recent node.
By the definition of recent, there exists an active node~$y$ in the subtree of~$x$
  such that $\level(x)\ge\level(y)-1$.
Let $k$ be an integer such that $x=\parent^k(y)$,
  and consider the path from~$y$ to~$x$, excluding~$x$:\;
  $\parent^0(y), \parent^1(y),\ldots,\parent^{k-1}(y)$.
None of these nodes is a fringe node:
  A fringe node would have to be in a different level than the active node~$y$,
  but that would force $\level(x)<\level(y)-1$.
We can thus prove by induction that all these nodes are not doomed:
  $\parent^0(y)$ is not doomed because it is active,
  and $\parent^{l+1}(y)$ is not doomed because $\parent^l(y)$ is not doomed and not fringe
    for $0<l<k$.
In fact, the induction from above also established that $x$~is not doomed.

We conclude that if a node is recent then it is not doomed.
\qed
\end{proof}


In the following we consider a sequence of $\ell$ modifying operations.
We write $\Rs{k}$ for the set of recent nodes after $k$~operations,
and $\Ms{k}$ for the set of nodes \emph{in memory} after $k$~operations, i.e., nodes that have been added but not (yet) deleted by the {\sf real-time} algorithm.

\begin{lemma} \label{lem-recent-are-in-memory}
For all $k \le \ell$:
\begin{itemize}
\item[(a)]
 We have $\Rs{k} \subseteq \Ms{k}$.
\item[(b)]
 If $\Ms{k} - \Rs{k} \ne \emptyset$, then the queue is nonempty
  after $k$ operations.
\end{itemize}
\end{lemma}
\begin{proof}
For point~(a),
  \autoref{lem:doomed-or-recent} together with the observation
  that only doomed nodes are scheduled for deletion suffice.
For point~(b),
  observe that the implementation uses a reference counting scheme
  that directly mirrors the definition of doomed nodes.
\qed
\end{proof}

\begin{lemma} \label{lem-invariant}
We have $|\Ms{k}| \le \max_{j \le k} |\Hs{j}_{< 2 h}|$ for all $k \le \ell$.
If the sequence is extensive then $|\Ms{k}| \le |\Hs{k}_{< 2 h}|$ for all $k \le \ell$.
\end{lemma}
\begin{proof}
We proceed by induction on~$k$.
The base case ($k = 0$) is trivial.
Let $0<k \le \ell$.
If $\Rs{k} = \Ms{k}$, then we have $\Ms{k} = \Rs{k} \subseteq \Hs{k}_{< 2 h}$ by \autoref{lem-recent-are-not-high}.
Hence $|\Ms{k}| \le |\Hs{k}_{< 2 h}|$.
By applying the induction hypothesis, it follows $|\Ms{k}| \le \max_{j \le k} |\Hs{j}_{< 2 h}|$.
So assume for the rest of the proof that the inclusion $\Rs{k} \subseteq \Ms{k}$ from \autoref{lem-recent-are-in-memory}~(a) is strict.
Then, by \autoref{lem-recent-are-in-memory}~(b), the queue is not empty after $k$~operations.
So the $k$th operation deletes from memory a node in the queue, and we have:
\begin{equation} \label{eq-lem-invariant}
|\Ms{k}| \le
\begin{cases}
|\Ms{k-1}|   & \text{if the $k$th operation is an $\addchild$} \\
|\Ms{k-1}|-1 & \text{if the $k$th operation is an $\deactivate$}
\end{cases}
\end{equation}
In either case we have $|\Ms{k}| \le |\Ms{k-1}|$.
By applying the induction hypothesis, it follows $|\Ms{k}| \le \max_{j \le k} |\Hs{j}_{< 2 h}|$.

Assume for the rest of the proof that the sequence is extensive.
Let the $k$th operation be an $\addchild$.
Then we have:
\[
|\Ms{k}| \ \mathop{\le}^\eqref{eq-lem-invariant} \ |\Ms{k-1}|
\ \mathop{\le}^\text{ind.~hyp.} \ |\Hs{k-1}_{< 2 h}|
\ \le \ |\Hs{k}_{< 2 h}|\;,
\]
where the last inequality is because no node is deactivated in the $k$th operation.
Let the $k$th operation be a $\deactivate$.
Then we have:
\[
|\Ms{k}| \ \mathop{\le}^\eqref{eq-lem-invariant} \ |\Ms{k-1}|-1
\ \mathop{\le}^\text{ind.~hyp.} \ |\Hs{k-1}_{< 2 h}|-1
\ \mathop{\le}^\text{\autoref{lem-extensive-nondecr}} \ |\Hs{k}_{< 2 h}|
\]
This concludes the proof.
\qed
\end{proof}

Now we can prove \autoref{thm-main}:

\begin{qtheorem}{\ref{thm-main}}
\stmtthmmain
\end{qtheorem}
\begin{proof}
 By combining Lemmas \ref{lem-invariant}~and~\ref{lem-heights-descending},
  $|\Ms{k}| \le 2 \max_{j \le k} |\Hs{j}_{<h}|$ for all $k \le \ell$.
 If the sequence is extensive then $|\Ms{k}| \le 2 |\Hs{k}_{<h}|$ for all $k \le \ell$.
 The theorem follows, as $\sgcs{k} \in \Theta\big( |\Hs{k}_{<h}| \big)$.
\qed
\end{proof}

}{} 
\end{document}